\documentclass[american,aps,prx,reprint,superscriptaddress, longbibliography,floatfix]{revtex4-2}
\usepackage[colorlinks=true, citecolor=red, urlcolor=blue ]{hyperref}
\usepackage{fullpage}
\usepackage{amsmath}
\usepackage{amssymb}
\usepackage{mathtools}
\usepackage{xcolor}
\usepackage{braket}
\usepackage{amsfonts}
\usepackage{amsthm}
\DeclareMathOperator{\tr}{tr}

\newtheorem{thm}{Theorem}
\newtheorem{corol}{Corollary}

\newtheorem{observation}{Observation}

\usepackage{array}
\usepackage{xcolor}
\usepackage{graphics}
\usepackage{tikz}

\newcommand{\ketbra}[2]{|#1\rangle \langle #2|}

\DeclarePairedDelimiter\abs{\lvert}{\rvert}%
\DeclarePairedDelimiter\norm{\lVert}{\rVert}%

\makeatletter
\let\oldabs\abs
\def\abs{\@ifstar{\oldabs}{\oldabs*}}
\let\oldnorm\norm
\def\norm{\@ifstar{\oldnorm}{\oldnorm*}}
\makeatother

\begin{document}


\title{Remote predictability of quantum measurement outcomes}

\author{Chirag Srivastava}
\affiliation{Institute of Informatics, National Quantum Information Centre, Faculty of Mathematics, Physics and Informatics, University of Gda\'nsk, Wita Stwosza 57, 80-308 Gda\'nsk, Poland}

\author{Aparajita Bhattacharyya}
\affiliation{Harish-Chandra Research Institute, A CI of Homi Bhabha National Institute, Chhatnag Road, Jhunsi, Prayagraj 211 019, India}

\author{Sheikh Parvez Mandal}\affiliation{Departamento de Física---CIOyN, Universidad de Murcia, E-30071 Murcia, Spain}
\author{Mahasweta Pandit}\affiliation{Departamento de Física---CIOyN, Universidad de Murcia, E-30071 Murcia, Spain}

\author{Ujjwal Sen}\affiliation{Harish-Chandra Research Institute, A CI of Homi Bhabha National Institute, Chhatnag Road, Jhunsi, Prayagraj 211 019, India}

\begin{abstract}
Predicting quantum measurement outcomes for a local observer makes sense before, but not after, the measurement. However, for a remote observer, predicting measurement outcomes even after measurement remains a valid question. We define remote predictability as the degree to which one observer can predict a measurement
outcome of a spatially separated observer, given full knowledge of the shared quantum state and
measurement setting. 
We show that the remote predictability before and after the measurement remains the same for product states, whereas it increases for all pure entangled states and for some classically correlated states. 
Perfect remote predictability for arbitrary projective measurements occurs only for maximally entangled states among all pure states, underscoring their special role. Comparing pure entangled states with their dephased versions, we find that dephasing on one subsystem can enhance remote predictability for a broad class of states and measurements - a counterintuitive, noise-induced advantage that vanishes for maximally entangled states under any projective measurement.

\end{abstract}

\maketitle
\section{Introduction}
Quantum theory is a well-established theoretical framework that successfully explains a wide range of nonclassical phenomena and paves the way for many quantum information processing tasks. The emergence of such phenomena and tasks can be attributed to the specific quantum states and measurements permitted within the formal structure of the theory. For instance, entangled quantum states~\cite{horodecki09,Guhne09}  are necessary for phenomena like Bell nonlocality~\cite{Bell64,bell_aspect_2004,bell-review} and quantum steering~\cite{Schrödinger_1935,Schrödinger_1936,Wiseman07,Uola20}. They also play a crucial role in numerous quantum information processing tasks, including quantum key distribution \cite{Ekert91, GRTZ02}, quantum teleportation \cite{Teleportation, Hu2023}, dense coding \cite{Densed1, SBDS19, Guo2019}, and device-independent randomness extraction \cite{Pironio2010, Christensen2013}.

A fundamental feature of quantum theory is that measurement outcomes are generally probabilistic. This leads to the idea of predictability of measurement outcomes before a measurement, based on their probability distribution. Predictability is also a quantum resource and has been used in wave-particle complementarity relations~\cite{GREENBERGER88, Basso2022}.  For single systems or systems without spatial separation between their subsystems, the outcome is completely known after a measurement; hence, the predictability of outcomes is trivial after the measurement. However, a different situation arises in a bipartite setup with spatially separated subsystems where the predictability holds non-trivial meaning both before and after the measurement. Consider two spatially separated parties, Alice and Bob, sharing a quantum state. Alice measures her part of the shared state, and Bob has complete knowledge of the shared quantum state and Alice’s measurement setting. Clearly, before Alice's measurement, the predictability for both parties is the same; however, after the measurement, Bob cannot know the outcome as Alice does, and therefore it is nontrivial to ask about Bob's prediction. 
Thus, in this work, we analyze \textit{remote predictability}, which captures the ability of one observer to infer the outcomes of measurements performed by a spatially separated observer. The (un)predictability of the measurement outcome of another party has been captured in only a limited amount of work~\cite{Martinez26, Berta10}. In~\cite{Martinez26}, the prediction of measurement outcomes is analyzed for measurements performed on two-qubit states using error measures from statistical learning theory. In contrast, the present work establishes the corresponding results for general finite-dimensional bipartite quantum states using state discrimination techniques. Similarly,~\cite{Berta10} investigates the uncertainty associated with the outcomes of two distinct measurement settings, whereas our focus is on the optimal prediction of the outcome of a \emph{single} measurement setting.

We observe that for a shared product state, the remote predictability remains the same before and after the measurement. However, this limitation can be overcome when the parties share either an arbitrary pure entangled state or some classically correlated states. The higher values of remote predictability for a wider range of measurement settings indicate a stronger correlation-generating strength of a quantum state. We show that perfect remote predictability can always be attained when an arbitrary projective measurement is performed on a subsystem in a maximally entangled state. This remarkable feature vanishes for arbitrary projective measurements on one-half of partially entangled states, highlighting the optimality of maximally entangled states. While it is natural to expect entanglement to play a key role in enabling remote prediction of measurement outcomes, it remains unclear whether entangled states always outperform classically correlated ones. To address this, we compare two scenarios: one involving a pure entangled state and another in which the same state undergoes dephasing~\cite{Devetak05, Breuer07, Wiseman09} on Alice’s subsystem. See Fig. \ref{fig:schematic} for the schematic of the scenarios. 
Surprisingly, our results show that, for a broad class of entangled states and certain measurement choices, the dephased (noisy) setup can outperform the noiseless case in terms of remote predictability.  However, this advantage disappears when the shared state is maximally entangled, and Alice performs any projective measurement. 
In this sense, maximally entangled states retain a fundamental advantage over a large set of pure entangled states with respect to remote predictability. 

\section{Remote predictability}
\label{sec:nonlocal_predictability}
Consider a scenario of two spatially separated parties, Alice and Bob, such that Bob prepares a composite system $AB$ in a quantum state, $\rho_{AB}\in \mathbb{C}^{d_A}\otimes\mathbb{C}^{d_B}$, sends subsystem $A$ to Alice, and keeps subsystem $B$ with himself. 
Alice performs a measurement on her share of the bipartite state after receiving the subsystem.  Let her measurement has $n \geq 2$ outcomes and is defined by a set of $n$ operators $\{\pi_i\geq0\}_{i=0}^{n-1}$ such that $\sum_i \pi_i =\mathbb{I}$, where $\mathbb{I}$ is an identity operator. 
In this setup, we assume that Alice's measurement setting is trusted (fixed and will not be altered) and also known to Bob. See Fig. \ref{fig:schematic}(a) for a schematic of the protocol. This measurement by Alice creates an ensemble, $\mathcal{E}_{\mathcal{\rho},\pi}$, at Bob's disposal, i.e.,
\begin{eqnarray}
\label{aisa}
   \mathcal{E}_{\mathcal{\rho},\pi}&=&\left\{p_i,\rho^i_B\right\}_{i=0}^{n-1},
\end{eqnarray}
where $p_i=\tr(\pi_i \otimes \mathbb{I}_B \rho_{AB})$, $\rho^i_{B}=\frac{\tr_{A}(\pi_i \otimes \mathbb{I}_{B} \rho_{AB})}{p_i}$, and $\mathbb{I}_X$ denotes the identity operator in the space of a system $X$. This means that Bob's subsystem is in state $\rho^i_B$ with probability $p_i$. This implies that Bob can guess Alice's outcome better if the states of the ensemble are more distinguishable. We refer to Bob's ability to predict the measurement outcome of Alice as the remote predictability. Thus, the remote predictability, $\mathcal{N}_n$, can be quantified by the success probability of minimum-error state discrimination of the states in the ensemble $\mathcal{E}_{\mathcal{\rho},\pi}$~\cite{helstrom1,Me2,Me4,Me6,Me7,Me12,Bae_2015,Halder20,Me16}, i.e.,
\begin{eqnarray}\label{eq:2}
    \mathcal{N}_n(\rho,\pi)&=&\max_{\{M_i\}}\sum_{i=0}^{n-1} p_i\tr(M_i \rho^i_B), \nonumber \\
    &=&\max_{\{M_i\}}\sum_{i=0}^{n-1} \tr\{M_i \tr_{A}(\pi_i \otimes \mathbb{I}_{B} \rho_{AB})\}
\end{eqnarray}
where the set of operators $\{M_i\geq 0\}_{i=0}^{n-1}$ such that $\sum_i M_i=\mathbb{I}_B$ is a measurement performed by Bob. The algebraic maximum that $\mathcal{N}_n$ can achieve is 1 because, from \eqref{eq:2}, we can see that it is the maximum of the convex combination of the probabilities $\tr(M_i\rho^i_B)$. We denote  $\mathcal{N}_n=1$ as the \emph{perfect} remote predictability. Note that $\mathcal{N}_n$ is invariant under any isometry on subsystem $B$. For the case where Alice performs a two-outcome measurement, Bob will have an ensemble of two states, and then the remote predictability is given by the Helstrom bound~\cite{helstrom1},
\begin{equation}
    \mathcal{N}_2=\frac{1+\Delta}{2}
\end{equation}
where $\Delta=\tr|p_1\rho^1_B-p_2\rho^2_B|$ and $\tr|O| = \tr(\sqrt{O^{\dagger} O})$ for any operator $O$. Since $\mathcal{N}_2$ is a linear and increasing function of $\Delta$, we will use $\Delta$ as a measure of remote predictability in the scenario when Alice performs a two-outcome measurement. 

\begin{figure}
    \centering
        \includegraphics[ trim= 5cm 17.5cm 8.4cm 4cm, clip, width=.9\linewidth]{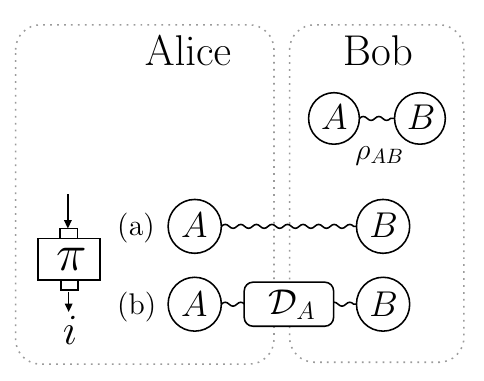}
    \caption{A schematic of the considered protocol. Bob prepares a state $\rho_{AB}$ and sends subsystem $A$ to Alice: (a) through a noiseless channel (Section~\ref{sec:nonlocal_predictability}), or (b) through a dephasing channel $\mathcal{D}_A$ (Section~\ref{sec:dephasing}). We investigate how well Bob can predict the outcomes of a measurement $\pi$ performed by Alice on $A$.
    }
    \label{fig:schematic}
\end{figure}

The remote predictability before measurement must correspond to the best guess of Bob about Alice's measurement outcomes before Alice performs the measurement. Given that Bob knows the shared state and the measurement setting of Alice, the remote predictability before the measurement is bounded by the probability of the most probable outcome of Alice's measurement.  Let that bound be denoted by $\mathcal{N}^{\mathrm{bef}}_n$, then
\begin{equation}
\mathcal{N}^{\mathrm{bef}}_n= \max_i\tr(\pi_i\otimes\mathbb{I}_B\rho_{AB}).
\end{equation}
Any violation of this bound suggests an intrinsic increase in the remote predictability after the measurement. That means an increase in Bob's ability to guess the outcomes of the measurement on the spatially separated system after the measurement is performed. Notice that this violation increases the guessing ability of Bob; however, based on his guess,  he cannot conclude whether the measurement by Alice has been performed or not. This is because the joint probability distribution of Alice and Bob's measurement outcomes is estimated on a quantum state, and the correlations in quantum theory satisfy the no-signaling principle~\cite{bell-review}.
For the scenario where the shared state is a product, i.e., $\rho_{AB}=\rho_A\otimes \rho_B$, Bob's remote predictability after the measurement is given by 
\begin{equation}
  \mathcal{N}_n= \max_{\{M_i\}}\sum_{i=0}^{n-1}\tr(\pi_i\rho_A)\tr\left(M_i\rho_B\right).
\end{equation}
 Now, since $\{\tr(M_i\rho_B)\}_i$ forms a probability distribution, it is easy to see that
\begin{eqnarray}\label{zaroor}
    \mathcal{N}_n&&\leq\max_{i}\tr(\pi_i\rho_A) \nonumber \\
    &&= \mathcal{N}^{\mathrm{bef}}_n.
\end{eqnarray}
Thus, Bob will not be able to guess Alice's measurement outcomes better even after Alice's measurement when they share a product state. However, in the following observation, we show that by sharing either an entangled state or a classically correlated separable state, it is possible to increase the remote predictability.  
\begin{observation}
    All pure entangled states and some classically correlated states increase the remote predictability after the measurement.
\end{observation}
\begin{proof}
Consider an arbitrary pure entangled state, expressed in Schmidt decomposition,
\begin{equation}
|\psi\rangle=\sum_{i=0}^{d-1}\sqrt{p_i}|i\chi_i\rangle,
\end{equation}
where $d$ is the dimension of subsystem $A$, $0<p_i<1$,  $p_i\geq p_{i+1}, \forall i$, $\sum_i p_i=1$,  $\{|i\rangle\}$ and $\{|\chi_i\rangle\}$ are sets of orthonormal states. Consider also a classically correlated state, $\rho=\sum_ip_i|i\chi_i\rangle\langle i\chi_i|$.
The remote predictability before measurement of Alice in the basis $\{|i\rangle\}$ on $|\psi\rangle$ and $\rho$ is given by $p_0<1$. But this measurement by Alice creates an ensemble of orthogonal states, $|\chi_i\rangle$ at the subsystem possessed by Bob, which are perfectly distinguishable. Thus, Bob can have a perfect remote predictability after measuring on his subsystem, surpassing the remote predictability bound before the measurement.  
\end{proof}
This establishes that any pure entangled state and some classically correlated states can be advantageous in the context of remote predictability. 

It is important to note that there exists a set of trivial measurements by Alice (at most one measurement operator acts non-trivially on subsystem $A$), given the Hilbert space of the shared state, which attain a perfect remote predictability for Bob.  For example, given the shared state such that subsystem $A$ lies in the support $\{|0\rangle,|1\rangle\}$,  the measurement: $\{\pi_0=|0\rangle\langle 0|+|1\rangle \langle 1|$, $\pi_1=|2\rangle\langle 2|\}$, where $|0\rangle$,~$|1\rangle$,~$|2\rangle$ are mutually orthogonal states, gives a perfect remote predictability.  Thus, it is interesting to examine remote predictability for \emph{non-trivial} measurements on subsystem $A$ which are defined as the measurements with at least two measurement operators acting non-trivially on the space of subsystem $A$.

It is easy to see from \eqref{zaroor} that any non-trivial projective measurement on one-half of a product state will yield an imperfect remote predictability. In the following two theorems, we show that any projective measurement on one-half of a pure entangled state yields a perfect remote predictability only when it is maximally entangled.

\begin{thm}\label{qayde}
     Any projective measurement on a subsystem of a maximally entangled state provides perfect remote predictability.
\end{thm}
\begin{proof}
Consider a maximally entangled state, \begin{equation} \label{hua}
    |\psi_m\rangle=\frac{1}{\sqrt{d}}\sum_{i=0}^{d-1}|i\chi_i\rangle,
\end{equation}
 where $d$ is the dimension of a subsystem, $\{|i\rangle\}$ and $\{|\chi_i\rangle\}$ are sets of orthonormal states. Let us assume that Alice performs a projective measurement in an arbitrary basis $\{|a_i\rangle\}$.  Since $|\psi_m\rangle$ is maximally entangled, it can be expressed as $\sum_{i}\frac{1}{\sqrt{d}}|a_i\eta_i\rangle$,  where $|a_i\rangle=U|i\rangle$ and $|\eta_i\rangle=V|\chi_i\rangle$, such that $U$ and $V$ are unitary operators. Clearly, when Alice measures in the basis $\{|a_i\rangle\}$, the ensemble of states appearing at Bob's part is a set of orthogonal states, $\{|\eta_i\rangle\}_{i=0}^{d-1}$,  with equal probability. Thus, the remote prediction of Bob is one. \\
The projective measurement was assumed to be of rank-1 till here; however, the same result can easily be extended to any rank projective measurements. Clearly, different outcomes of an arbitrary rank projective measurement by Alice on a maximally entangled state will create a set of states, lying in orthogonal spaces, at Bob's subsystem, and hence the remote predictability will still be one.  
\end{proof}
Note that a consequence of the above result is that the outcomes of any two measurement settings can be predicted precisely with the aid of quantum memory, as discussed in ~\cite{Berta10}. For a general pure entangled state,
we will see that the perfect remote predictability does not hold for arbitrary projective measurements. 
Now consider an arbitrary pure entangled state, expressed in its Schmidt decomposition, i.e.,
\begin{equation}\label{diwar}
    |\psi\rangle=\sum_{i=0}^{d-1}\sqrt{p_i}|i\chi_i\rangle,
\end{equation}
where $p_i\geq 0,~\sum_ip_i=1,$ and $\{|i\rangle\}_{i=0}^{d-1}$ and $\{|\chi_i\rangle\}_{i=0}^{d-1}$ form orthonormal bases for systems $A$ and $B$, respectively. We refer to $\{|i\rangle\}$ as a Schmidt basis of system $A$ for  $|\psi\rangle$.  
Assume that Alice measures her subsystem in an arbitrary basis $\{|a\rangle\}_{a=0}^{d-1}$. Now we can express Schmidt basis, $\{|i\rangle\}$, in terms of the arbitrary basis, i.e.,
\begin{equation}\label{beghar}
    |i\rangle=\sum_a \alpha^i_a|a\rangle, ~~\forall ~i,
\end{equation}
where $\sum_a|\alpha^i_a|^2=1$ and $\sum_a\alpha^i_a (\alpha^{j}_a)^*=0,~\forall i\neq j,$ due to the orthonormality of the Schmidt basis. We can also write,
\begin{equation}
    |a\rangle=\sum_i (\alpha^i_a)^*|i\rangle, ~~\forall ~a,
\end{equation}
where $\sum_i|\alpha^i_a|^2=1$ and $\sum_i(\alpha^i_a)^* \alpha^{i}_{a'}=0,~\forall a\neq a',$ due to the orthonormality of $\{|a\rangle\}$. 
Now it is easy to see that Bob has the ensemble of states, 
\begin{equation}
|\phi_a\rangle=\sum_i\frac{\sqrt{p_i}\alpha^i_a}{\sqrt{\sum_i p_i |\alpha^i_a|^2}}|\chi_i\rangle,
\end{equation}
with the probability $q_a=\sum_i p_i |\alpha^i_a|^2$, when Alice gets the outcome $|a\rangle$ after her measurement. Now, 
\begin{equation}\label{haara_nahi}
    \langle \phi_a|\phi_{a'}\rangle=\sum_i p_i (\alpha^i_a)^*\alpha^i_{a'}.
\end{equation}
For maximally entangled states, $p_i=p_j=\frac{1}{d}$, and thus $\langle \phi_{a'}|\phi_{a}\rangle=\frac{1}{d}\sum_i(\alpha^i_a)^* \alpha^{i}_{a'}=0,~\forall a\neq a'$. Thus, Bob has an ensemble of orthogonal states, implying perfect remote predictability for an arbitrary projective measurement by Alice on her share of a maximally entangled state. This also presents a parallel proof for Theorem \ref{qayde}. 

\begin{thm}
    There always exists a projective measurement on one-half of a partially entangled state, yielding an imperfect remote predictability.
\end{thm}
\begin{proof} 
For a partially entangled state, there exists at least one pair, $i_0\in\{0,1,\ldots,d-1\}$ and $j_0\in\{0,1,\ldots,d-1\}$, such that $p_{i_0}\neq p_{j_0}$ in the structure of the state given in Eq. \eqref{diwar}. Without loss of generality, we set $i_0$ as $0$ and $j_0$ as $1$. Consider projective measurement for Alice in the basis, $\{\frac{1}{\sqrt{2}}|0+1\rangle,\frac{1}{\sqrt{2}}|0-1\rangle, \{|i\rangle\}_{i=2}^{d-1}\}.$ After her measurement, Bob has the ensemble of states, $$\left\{\frac{\sqrt{p_0}|0\rangle+\sqrt{p_1}|1\rangle}{\sqrt{p_0+p_1}},\frac{\sqrt{p_0}|0\rangle-\sqrt{p_1}|1\rangle}{\sqrt{p_0+p_1}}, \{|i\rangle\}_{i=2}^{d-1}\right\},$$
with probability $\{\frac{p_0+p_1}{2},\frac{p_0+p_1}{2}, \{p_i\}_{i=2}^{d-1}\}$.
Clearly, these states are not mutually orthogonal since $p_0\neq p_1$. Thus, these states can never be perfectly discriminated, and hence remote predictability is not perfect.
\end{proof}

\begin{observation}\label{Ehsaan}
A partially entangled two-qubit state never yields a perfect remote predictability
for non-trivial projective measurements in an arbitrary basis except the Schmidt basis.
\end{observation}
\begin{proof} 
See Appendix \ref{app:A0}.
\end{proof}

\section{Pure entangled state versus its dephased counterpart}
\label{sec:dephasing}
In this section we aim to compare the remote predictability of a measurement on one-half of a pure entangled state with the same measurement on the one-half of a classically correlated state. For a valid comparison with a pure entangled state, it is motivating to choose a classically correlated state which turns out after a noisy map acts on the entangled state. Since we assume that Bob prepares the joint system and sends one subsystem to Alice, it is plausible to assume that Alice's subsystem is influenced by some noise. We assume here that such a noise is a dephasing map and see its effect on the remote predictability of Bob (see Fig. \ref{fig:schematic}(b) for a schematic). Another reason to consider dephasing on Alice's system is that it allows us to compare a coherent superposition of product states (a pure entangled state) with a classical mixture of the same product states (dephased counterpart of the pure entangled state). Without loss of generality, we assume that the systems undergo dephasing in the computational basis. A full dephasing map acting on Alice's part of the shared state, $\rho_{AB}$, is defined as:
\begin{equation}\label{ruswa}
\mathcal{D}_A(\rho_{AB}):=\sum_i|i_A\rangle\langle i_A|\langle i_A|\rho_{AB}|i_A\rangle,
\end{equation}
 where $\{|i_A\rangle\}$ forms a computational basis for system $A$. In the following, $\mathcal{D}_A(\rho_{AB})$ will be referred to as the dephased counterpart of $\rho_{AB}$.
Any pure state $|\psi_{AB}\rangle$, up to a local unitary on Bob's part, can be expressed as $\sum_i \sqrt{p_i}|i\chi_i\rangle$, where $\{|\chi_i\rangle\}$ forms a basis for Bob's part but need not be mutually orthogonal.
Notice that  $\mathcal{D}_A(|\psi_{AB}\rangle\langle \psi_{AB}|)=\sum_ip_i |i\chi_i\rangle\langle i\chi_i|$, i.e., the dephased state of a pure state, expressed as a superposition of certain product states, is in a mixture of the same product states.

Now we compare the action of no dephasing with the full dephasing acting on Alice's part on Bob's remote predictability. If the state is unaffected by any noise, then the ensemble created at Bob's part is:
\begin{equation}
     \mathcal{E}_{\mathcal{\rho}}=\left\{s_i,~\frac{\tr_{A}(\pi_i \otimes \mathbb{I}_B \rho_{AB})}{s_i}\right\}_{i=0,1,\ldots,n-1}, \nonumber
\end{equation}
where $\{s_i=\tr(\pi_i \otimes \mathbb{I}_B \rho_{AB})\}_{i=0,1,\ldots,n-1}$. On the other hand, if the state undergoes full dephasing on Alice's part, then Bob's ensemble is:
\begin{eqnarray}
\label{ens1} \mathcal{E}_{\mathcal{D}_A(\rho)}=\left\{m_i,~\frac{\tr_{A}(\pi_i \otimes \mathbb{I}_B \mathcal{D}_A(\rho_{AB}))}{m_i}\right\}_{i=0,1,\ldots,n-1}, \nonumber
\end{eqnarray}
where $\{m_i=\tr(\pi_i \otimes \mathbb{I}_B \mathcal{D}_A(\rho_{AB}))\}_{i=0,1,\ldots,n-1}$.

\begin{observation}\label{dusra}
    There always exists a non-trivial measurement by Alice such that the remote predictability produced by a pure entangled state can always be produced by its dephased counterpart.
\end{observation}
\begin{proof}
    $\mathcal{N}_n$ is invariant under the action of a unitary on Bob's subsystem. Thus, any pure entangled state, up to local unitary on $B$, can be expressed as $|\psi_{AB}\rangle=\sum_{i}\sqrt{p_i}|i_A\eta^i_B\rangle$, where $\{|i_A\rangle\}_{i=0}^{d_A-1}$  is the computational basis on the subsystem $A$ and $\{|\eta^i_B\rangle\}_{i=0}^{d_B-1}$ is a basis on the subsystem $B$ (elements of this basis may not necessarily be orthogonal).  Now consider another state $\mathcal{D}_A (|\psi_{AB}\rangle\langle\psi_{AB}|)=\sum_i p_i |i\eta^i\rangle\langle i\eta^i|.$  Let Alice perform the measurement on the computational basis. Then, from the definition of remote predictability, we have,
    \begin{equation}
        \mathcal{N}_n(\mathcal{D}_A(|\psi_{AB}\rangle\langle\psi_{AB}|), \{|i\rangle\})=\mathcal{N}_n(|\psi_{AB}\rangle\langle\psi_{AB}|, \{|i\rangle\}). \nonumber
    \end{equation}
\end{proof}
Therefore, for computational basis measurement by Alice,  the outcomes are equally remotely predictable by Bob when the shared state is a pure entangled state or when the same state has undergone a full dephasing on Alice's subsystem. Thus, entangled states are not special in the context of remote predictability if Alice measures on a certain basis. In the following theorem, we show the existence of measurements and a large set of pure entangled states such that the dephased counterpart of a pure entangled state from the set can have greater remote predictability than the pure state itself.

\begin{thm}
\label{th1}
There exist measurements by Alice such that the remote predictability is greater when Alice and Bob share a state, $\mathcal{D}_A(\rho_{AB})$, defined in eq. \eqref{ruswa}, compared to when they share state, $\rho_{AB}$, for a large set of pure entangled states, $\rho_{AB}$. 
\end{thm}

\begin{proof} We prove the statement by showing that there exists at least one measurement, for a large set of bipartite pure entangled states, $|\psi_0\rangle_{AB}$, such that $\Delta_{\mathcal{D}_A(\psi_0)}-\Delta_{\psi_0}>0$, where $\psi_0:=\ketbra{\psi_0}{\psi_0}_{AB}$. See Appendix \ref{app:A1} for the complete proof. 
\end{proof}

Thus, dephasing of Alice's system can be advantageous over no dephasing on a large set of pure entangled states, in the context of remote predictability. 
This leads to the counterintuitive conclusion that a mixture of product states, obtained by dephasing a coherent superposition of the same product states, can yield higher remote predictability than the coherent superposition itself. This implies that the noise on Alice's system might help Bob in predicting Alice's measurement outcomes.
However, from the proof of Theorem \ref{qayde}, we observe that the maximally entangled states do not fall in the broad class of states such that there exists a projective measurement providing greater remote predictability for the dephased counterpart of a maximally entangled state over itself. 
In the following, we prove that maximally entangled states provide greater remote predictability than their dephased counterpart if Alice performs a non-trivial projective measurement in a basis different than the computational basis.

\begin{thm}\label{dhun}
For any non-trivial projective measurement, apart from the measurement in the computational basis, on one-half of the maximally entangled state will provide greater remote predictability than its fully dephased counterpart. 
\end{thm}
\begin{proof}
Consider the state $|\psi_m\rangle$, given in Eq. \eqref{hua}, then $\mathcal{D}_A(\psi_m)=\frac{1}{d}\sum_{i=0}^{d-1} |i\chi_i\rangle\langle i\chi_i|$. Let Alice perform a rank-1 projective measurement in the basis $\{|a_i\rangle\}$. 
Let us expand $|i\rangle$ in the basis $\{|a_j\rangle\}$, i.e., $|i\rangle=\sum_j\alpha^i_j|a_j\rangle$, where $\alpha^i_j$ are complex coefficients. Now it is not difficult to see that after the Alice's measurement in the basis $\{|a_j\rangle\}$, Bob has the ensemble of states, $$\Bigg\{\frac{\sum_i |\alpha^i_j|^2 |\chi_i\rangle\langle \chi_i|}{\sum_i |\alpha^i_j|^2}\Bigg\}_{j=0}^{d-1}$$ with probability $\frac{1}{d}\sum_i|\alpha^i_j|^2$. 
But these states are orthogonal only if $\{|a_i\rangle\}$ is same as $\{|i\rangle\}$ up to some permutation of the states in $\{|i\rangle\}$. Thus, the remote predictability will be less than one if Alice performs a rank-1 projective measurement in a basis $\{|a_i\rangle\}$ different from the basis $\{|i\rangle\}$.   \\
Now, if the states at Bob's end are not orthogonal for a rank-1 projective measurement in the basis $\{|a_i\rangle\}$, then they must not be orthogonal for a higher-rank projective measurement in the same basis. Thus, the remote predictability will be less than one. But, on the other hand, remote predictability is perfect using projective measurements on maximally entangled states. Hence, the remote predictability is greater for a maximally entangled state than its dephased state when Alice performs any projective measurement in a basis which is different from the basis used for dephasing. 
\end{proof}

\begin{corol}
    The remote predictability is the same for the maximally entangled state and its dephased counterpart if the projective measurement is performed in the computational basis.
\end{corol}
\begin{proof}
The proof of Theorem \ref{dhun}, also shows that if Alice performs measurement in the computational basis $\{|i\rangle\}$, the same basis used to dephase $|\psi_m\rangle$, given in Eq. \eqref{hua}, then the ensemble of states at Bob's end becomes orthogonal and thus the remote predicatbility is one even when the shared state is $\mathcal{D}_A(\psi).$
\end{proof}

The results of theorems \ref{th1} and \ref{dhun}, indicate that the maximally entangled state may have an edge over other pure entangled states in the context of remote predictability.
Now consider a more general dephasing map on Alice's system which is defined as:
 \begin{equation}\label{tanha}
\mathcal{P}^q_A(\rho_{AB}):=q\mathcal{D}_A(\rho_{AB})+(1-q)\rho_{AB},
 \end{equation}
where $0\leq q\leq 1$.
In the following theorem, we show that the maximally entangled state is advantageous for remote predictability compared to the state derived from \emph{any amount} of dephasing to the maximally entangled state for any two-outcome projective measurement. 

\begin{thm}
\label{th2}
There are no two-outcome projective measurements on subsystem $A$ such that the remote predictability is higher for $\mathcal{P}^q_A(\rho_{AB})$, defined in eq. \eqref{tanha}, than the remote predictability for the maximally entangled state $\rho_{AB}$. 
\end{thm}

\begin{proof} See Appendix \ref{app:A2}.
\end{proof}

Note that the choice $q=1$ reduces the result of the theorem to the scenario where a complete dephasing channel acts on Alice's subsystem.  Thus, remote predictability for any projective measurement by Alice never increases when Alice's system of the shared maximally entangled state undergoes any amount of dephasing.


\section{Conclusion}
~We studied remote predictability, an operational quantity that characterizes how well one observer can infer the outcome of a measurement performed by a spatially separated observer on a shared quantum system.    We showed that product states offer no advantage for remote prediction beyond the information available prior to measurement, whereas arbitrary pure entangled states and certain classically correlated states can enhance predictability.  
We proved that maximally entangled states possess a remarkable and universal property: perfect remote predictability can always be achieved for arbitrary projective measurements performed on one subsystem. In contrast, this property is generally absent for partially entangled states. 
Next, we found that dephasing on one subsystem can, in certain regimes, enhance remote predictability. Thus, while decoherence is generally believed to degrade quantumness~\cite{zurek2003, schlosshauer2007}, our results highlight a complementary regime where controlled noise can aid remote predictability.
However, this enhancement vanishes for maximally entangled states under any projective measurement, indicating that such states retain a fundamental optimality in this context.
The observed enhancement of remote predictability under dephasing echoes a broader class of noise-assisted quantum effects, e.g., environment-assisted quantum transport in molecular networks~\cite{Plenio2008, Rebentrost2009, Mohseni2008, Caruso2009} and entanglement-generation through dissipative state engineering~\cite{Verstraete2009, krauter2011}. The framework for remote predictability may also inspire new protocols for distributed remote predictability in quantum networks~\cite{Branciard10, Branciard12, Fritz12, Boreiri23} and sequential scenarios~\cite{Silva15, Curchod17, Brown20, SPS22, PSS22, SPS25, MSSS25}, and can be useful in quantum channel determination and discrimination tasks~\cite{Matthews10, Holevo12, Piani15}.

\acknowledgments
C.S. acknowledges partial support by the National Science Centre, Poland, under grant Opus 25, 2023/49/B/ST2/02468. A.B. acknowledges support from ‘INFOSYS scholarship for senior students’ at Harish Chandra Research Institute, India. MP and SPM were supported by the QuantERA II program (Mf-QDS) and QuantERA III program (AQuSeND) that have received funding from the European Union’s Horizon 2020 research and innovation program under Grant Agreement No.101017733 and from the Agencia Estatal de Investigación, with project codes PCI2022-132915, PCI2024-153474 and by QNAVIUM Project SCPP2400C011413XV0 funded by MICIU/AEI/10.13039/501100011033 and by the European Union NextGenerationEU/PRTR, and by the European Union (Quantum Flagship project ASPECTS, Grant Agreement No. 101080167). US acknowledges financial support from the Anusandhan National Research Foundation (ANRF), Government of India, under the Grant No. ANRF/ARG/2025/004617/PS.

\onecolumngrid
\section*{Appendix}
\subsection{Proof of Observation  \ref{Ehsaan}}\label{app:A0}
From Eq. \eqref{diwar}, we have $d=2$ and $p_0\neq p_1$ for two-qubit partially entangled pure states. 
We assume that Alice performs a non-trivial projective measurement, i.e., at least two projectors of the measurement act non-trivially on Alice's qubit. Thus, the measurement can be represented by a couple of orthonormal states $\{|a\rangle\}_{a=0}^1$, which can be expressed in the computational basis, $\{|i\rangle\}_{i=0}^1$ for a qubit, as given in 
\eqref{beghar}. We also consider it to be different from the Schmidt basis $\{|i\rangle\}$. These facts mean $\alpha^i_a\neq 0,~\forall i,a.$ 
From Eq. \eqref{haara_nahi}, the inner product between the two states of Bob's ensemble after Alice's measurement,
\begin{eqnarray}
    \langle \phi_0|\phi_1\rangle&=&p_0(\alpha^0_0)^* \alpha^0_1+p_1(\alpha^1_0)^* \alpha^1_1 \nonumber \\
    &=&p_0(\alpha^0_0)^* \alpha^0_1-p_1(\alpha^0_0)^* \alpha^0_1 \nonumber \\
    &=&(p_0-p_1)(\alpha^0_0)^* \alpha^0_1,
\end{eqnarray}
where in the second equality, we use the identity $\sum_i(\alpha^i_a)^* \alpha^{i}_{a'}=0,~\forall a\neq a'$. 
Since $\alpha^i_a\neq 0,~\forall i,a$ and $p_0\neq p_1$, thus $\langle \phi_0|\phi_1\rangle\neq 0$.
This implies that the states at Bob's end are not mutually orthogonal and hence will never have perfect distinguishability. We thus conclude that an arbitrary two-outcome non-trivial projective measurement, except in the Schmidt basis, on one-half of partially entangled two-qubit states never yields perfect remote predictability.

\subsection{Proof of Theorem \ref{th1}}\label{app:A1}
Since the quantity $\Delta_\psi$ for any state $\psi_{AB}$ is invariant if any unitary acts on the system $B$, any result for an arbitrary pure state up to a local unitary on system $B$ is also true for an arbitrary pure state. An arbitrary pure bipartite entangled state, up to a local unitary on system $B$, can be expressed as
\begin{eqnarray}
\label{superposed}
    \ket{\psi}_{AB}=\sum_{i=0}^{d-1}\sqrt{p_i}\ket{i\chi_i},
\end{eqnarray}
where $\sum_ip_1=1$, $p_i> 0$ $\braket{i|j}=\delta_{ij}$, $d\geq 2$ is the dimension of system $A$, $\{\ket{\chi_i}\}$ are normalised but not necessarily orthogonal, such that $|\chi_i\rangle=\sum_{j=0}^i t_{ij} |j \rangle$, where $\sum_j|t_{ij}|^2=1$ $\forall j=0,1,\ldots i$, $\forall i=0,1,\ldots,d-1$. The dephased state, $\mathcal{D}_A(\psi)=\sum_i p_i |i\chi_i\rangle\langle i\chi_i|$. 
For simplicity and without loss of generality, we explicitly define some $t_{ij}$'s: $t_{00}=1$, $t_{10}=c_\theta$, $t_{11}=s_\theta e^{i\phi}$, where 
$c_\theta:=\cos\frac{\theta}{2}$, $s_\theta:=\sin\frac{\theta}{2}$, $0\leq \theta\leq \pi$, and  $0\leq \phi \leq 2\pi$.

Consider a two-outcome projective measurement:
\begin{eqnarray}\label{dekha}
\pi_0&=&P(c_\beta|0\rangle+s_\beta e^{i\gamma}|1\rangle), \nonumber \\ \pi_1&=&P(s_\beta|0\rangle-e^{i\gamma}c_\beta|1\rangle),
\end{eqnarray}
where $0\leq\beta\leq\pi$, $0\leq\gamma<2\pi$, and $P(|\nu\rangle)=|\nu\rangle\langle\nu|$. 
Let $\sigma=\pi_0-\pi_1$, then $\sigma=c_\beta|0\rangle\langle0|-c_\beta|1\rangle\langle 1| + s_ \beta(e^{-i\gamma}|0\rangle\langle1|+e^{i\gamma}|1\rangle\langle0|)$.

Now,
\begin{eqnarray}
\Delta_\psi&=&\tr|\tr_{A}\left(\sigma \otimes \mathbb{I} \psi\right)|, \\    \Delta_{\mathcal{D}_A(\psi)}&=&\tr|\tr_{A}\left(\sigma \otimes \mathbb{I} \mathcal{D}_A(\psi)\right)|.
\end{eqnarray}

Let's first compute $\Delta_\psi$:
    \begin{eqnarray}
    \tr_{A}\left(\sigma \otimes \mathbb{I} \psi\right)&=&p_0\sigma_{00}|0\rangle\langle 0|+p_1\sigma_{11}|\chi_1\rangle\langle \chi_1|+\sqrt{p_0p_1}\left(\sigma_{01}|0\rangle\langle \chi_1|+\sigma_{10}|\chi_1\rangle\langle 0|\right) \nonumber \\
    &=& p_0\sigma_{00}|0\rangle\langle 0|-p_1\sigma_{00}|\chi_1\rangle\langle \chi_1|+\sqrt{p_0p_1}\left(\sigma_{01}|0\rangle\langle \chi_1|+\sigma_{10}|\chi_1\rangle\langle 0|\right), \nonumber \\
    &=& \left(p_0\sigma_{00}-p_1\sigma_{00}c_\theta^2+\sqrt{p_0p_1}(\sigma_{10}+\sigma_{01})c_\theta\right)|0\rangle\langle 0| -p_1\sigma_{00}s_\theta^2|1\rangle\langle 1|\nonumber \\&&+\left\{\left(-p_1\sigma_{00}c_\theta +\sqrt{p_0p_1}\sigma_{01}\right)s_\theta e^{-i\phi}|0\rangle\langle 1| + h.c.\right\},
\end{eqnarray}

where $\sigma_{ij}=\langle i|\sigma|j\rangle$ and second line above follows from the fact that $\sigma=\pi_0-\pi_1$ is a trace-less operator. Also, $O+h.c.=O + O^\dagger$ for any operator $O$.
Let the two eigenvalues of the operator $\tr_{A}\left(\sigma \otimes \mathbb{I} \psi\right)$ are denoted by $\epsilon_1$ and $\epsilon_2$ ($\epsilon_1 \geq \epsilon_2$). Thus, $\Delta_S=|\epsilon_1|+|\epsilon_2|$.  
Now the determinant of the matrix $\tr_{A}\left(\sigma \otimes \mathbb{I} \psi\right)$ is $-p_0p_1s_\theta^2(\sigma_{00}^2+|\sigma_{01}|^2)$, which is negative and thus $\epsilon_1$ and $\epsilon_2$ have opposite signs. This implies 
    \begin{eqnarray}
    \Delta_\psi&=&\epsilon_1-\epsilon_2, \nonumber \\
    &=& \Big[ (p_1-p_0)^2\sigma_{00}^2+4p_0p_1s_\theta^2\sigma_{00}^2-2(p_1-p_0)\sqrt{p_0p_1}\sigma_{00}(\sigma_{10}+\sigma_{01})c_\theta \nonumber \\ &&+p_0p_1(\sigma_{10}+\sigma_{01})^2c_\theta^2+4p_0p_1|\sigma_{01}|^2s_\theta^2\Big]^\frac{1}{2}. 
\end{eqnarray}
Now let's compute $\Delta_{\mathcal{D}_A(\psi)}$:
\begin{eqnarray}
    \tr_{A}\left(\sigma \otimes \mathbb{I} {\mathcal{D}_A(\psi)}\right)&=&p_0\sigma_{00}|0\rangle\langle 0|-p_1\sigma_{00}|\psi\rangle\langle \psi| \nonumber \\
    &=&\left(p\sigma_{00}-p_1\sigma_{00}c_\theta^2\right)|0\rangle\langle 0|-p_1\sigma_{00}s_\theta^2|1\rangle\langle 1| \nonumber \\
    &&-p_1\sigma_{00}c_\theta s_\theta \left(e^{i\phi}|0\rangle\langle 1| + h.c.\right).
\end{eqnarray}

Let the two eigenvalues of the operator $\tr_{A}\left(\sigma \otimes \mathbb{I} \mathcal{D}_A(\psi)\right)$ are denoted by $\epsilon^M_1$ and $\epsilon^M_2$ ($\epsilon^M_1 \geq \epsilon^M_2$). Now the determinant of the matrix $\tr_{A}\left(\sigma \otimes \mathbb{I} {\mathcal{D}_A(\psi)}\right)$ is $-p_0p_1s_\theta^2\sigma_{00}^2$, which is negative and thus $\epsilon^M_1$ and $\epsilon^M_2$ have opposite signs. This implies
\begin{eqnarray}
   \Delta_{\mathcal{D}_A(\psi)}&=&\epsilon^M_1-\epsilon^M_2, \nonumber \\
   &=& \Big[(p_1-p_0)^2\sigma_{00}^2+4 p_0p_1\sigma_{00}^2s_\theta^2\Big]^\frac{1}{2}.
\end{eqnarray}
Now $\Delta_{\mathcal{D}_A(\psi)}^2>\Delta_\psi^2$ if
    \begin{eqnarray}\label{ehem}
   p_0p_1(\sigma_{10}+\sigma_{01})^2c_\theta^2+4p_0p_1|\sigma_{01}|^2s_\theta^2 
     < 2(p_1-p_0)\sqrt{p_0p_1}\sigma_{00}(\sigma_{10}+\sigma_{01})c_\theta
\end{eqnarray}
Now we will show that there exists a range of parameters for the choice of measurement, given in Eq. \eqref{dekha}, and a large set of pure entangled states, such that the above condition will be satisfied.
Thus the condition $\Delta_{\mathcal{D}_A(\psi)}^2>\Delta_\psi^2$ becomes
  \begin{eqnarray}\label{jaha}
    4p_0p_1\sin^2\beta(\cos^2\gamma c_\theta^2 +s_\theta^2) <4(p_1-p_0)\sqrt{p_0p_1}\cos\beta\sin\beta\cos\gamma c_\theta.
\end{eqnarray}  
For the parameters of states: $p_0\neq p_1$ or $\theta\neq\pi$), we will show that $\Delta_{\mathcal{D}_A(\psi)}>\Delta_\psi$ for some measurements. Let's choose the measurement such that $\gamma=0$ and $\beta \neq 0$. Then inequality \eqref{jaha} boils down to
\begin{eqnarray}\label{dhua}
    4\sqrt{p_0p_1}\sin\beta \Big(\sqrt{p_0p_1}\sin\beta
    -(p_1-p_0)\cos\beta\cos\frac{\theta}{2}\Big)<0,
\end{eqnarray}
and since $4\sqrt{p_0p_1}\sin\beta>0$, this implies
\begin{equation}\label{kua}
    \tan\beta<\frac{(p_1-p_0)\cos\frac{\theta}{2}}{\sqrt{p_0p_1}},
\end{equation}
which can always be satisfied since $\tan\beta \in (-\infty, \infty)$. 

\hfill $\blacksquare$

\subsection{Proof of Theorem \ref{th2}}\label{app:A2}

Consider maximally entangled states, which are of the form $\rho_{AB}=\ketbra{\psi}{\psi}$ with $\ket{\psi}=\sum_{i=1}^d\sqrt{p_i}\ket{i\chi_i}$, where $\braket{i|j}=\braket{\chi_i|\chi_j}=\delta_{ij}$,  and $p_i=1/d$, $\forall i,j$.  Let the two-outcome projective measurement be given by: $\{\pi_0,\pi_1\}$,  $\pi_0+\pi_1=\mathbb{I}_r$. Without loss of generality, we assume $r\leq d$ and let $\sigma:=\pi_0-\pi_1$.
For such states and measurements, 
\begin{eqnarray}
\label{d_dr}
    \Delta_{\mathcal{P}^q_A(\rho)} &=& \frac{1}{d}\tr \Big| (1-q)\sum_{\alpha \beta} \bra{\beta}\sigma\ket{\alpha}\ketbra{\chi_{\alpha}}{\chi_{\beta}} 
    + q\sum_i \bra{i}\sigma\ket{i}\ketbra{\chi_i}{\chi_i} \Big| \nonumber \\
    &\le& \frac{1-q}{d} \tr \Big| \sum_{\alpha \beta} \bra{\beta}\sigma\ket{\alpha}\ketbra{\chi_{\alpha}}{\chi_{\beta}}\Big| 
    + \frac{q}{d}\tr \Big|\sum_i \bra{i}\sigma\ket{i}\ketbra{\chi_i}{\chi_i} \Big| \nonumber \\
    &\le&\frac{1-q}{d} \tr\sqrt{\sum_{\substack{\alpha \beta  \\ \alpha' }}    
     \bra{\alpha'}\sigma\ket{\beta}\bra{\beta}\sigma\ket{\alpha} \ketbra{\chi_{\alpha}}{\chi_{\alpha'}}}  
    + \frac{q}{d}\tr\left(\sum_i \Big|\bra{i}\sigma\ket{i}\Big| \ketbra{\chi_i}{\chi_i} \right) \nonumber \\
    &=&  \frac{(1-q)}{d} \frac{r}{d} + \frac{q}{d} \sum_i \left|\bra{i}\sigma\ket{i}\right| \le \frac{r}{d}, 
\end{eqnarray}
where the first and second inequalities use the triangle law for the trace norm, and the third inequality is due to the fact that $|\langle i | \sigma|i\rangle|\leq 1$ for each $i$.
Whereas the quantity $\Delta_{\rho}$, is given by
\begin{eqnarray}
\label{d_r}
    \Delta_{\rho}&=&\frac{1}{d}\tr \left| \sum_{\alpha \beta} \bra{\beta}\sigma\ket{\alpha}\ketbra{\chi_{\alpha}}{\chi_{\beta}} \right| \nonumber \\
    &=& \frac{1}{d} \tr\left[\left( \sum_{\alpha \beta \alpha'} \bra{\alpha'}\sigma\ket{\beta}\bra{\beta}\sigma\ket{\alpha}\ketbra{\chi_{\alpha}}{\chi_{\alpha}'} \right)^{1/2}\right] \nonumber \\
    &=& \frac{1}{d} \tr\left[\left( \sum_{\alpha  \alpha'} \bra{\alpha'}\sigma^2\ket{\alpha}\ketbra{\chi_{\alpha}}{\chi_{\alpha}'} \right)^{1/2}\right] \nonumber \\
    &=& \frac{r}{d}. 
\end{eqnarray}
Therefore, from equations~\eqref{d_dr} and~\eqref{d_r}, we finally obtain $\Delta_{\mathcal{P}^q_A(\rho)}-\Delta_{\rho} \le 0$. 
\hfill $\blacksquare$

\twocolumngrid
\bibliography{references}

@article{Plenio2008,
  title={Dephasing-assisted transport: quantum networks and biomolecules},
  author={Plenio, Martin B and Huelga, Susana F},
  journal={New Journal of Physics},
  volume={10},
  number={11},
  pages={113019},
  year={2008},
  publisher={IOP Publishing},
    doi={10.1088/1367-2630/10/11/113019}
}

@article{Rebentrost2009,
  title={Environment-assisted quantum transport},
  author={Rebentrost, Patrick and Mohseni, Masoud and Kassal, Ivan and Lloyd, Seth and Aspuru-Guzik, Al{\'a}n},
  journal={New Journal of Physics},
  volume={11},
  number={3},
  pages={033003},
  year={2009},
  publisher={IOP Publishing},
  doi={10.1088/1367-2630/11/3/033003}
}

@article{Mohseni2008,
  title={Environment-assisted quantum walks in photosynthetic energy transfer},
  author={Mohseni, Masoud and Rebentrost, Patrick and Lloyd, Seth and Aspuru-Guzik, Alan},
  journal={The Journal of chemical physics},
  volume={129},
  number={17},
  year={2008},
  publisher={AIP Publishing},
url={https://pubs.aip.org/aip/jcp/article-abstract/129/17/174106/187931/Environment-assisted-quantum-walks-in?redirectedFrom=fulltext}
}

@article{Caruso2009,
  title={Highly efficient energy excitation transfer in light-harvesting complexes: The fundamental role of noise-assisted transport},
  author={Caruso, Filippo and Chin, Alex W and Datta, Animesh and Huelga, Susana F and Plenio, Martin B},
  journal={The Journal of Chemical Physics},
  volume={131},
  number={10},
  year={2009},
  publisher={AIP Publishing},
url={https://pubs.aip.org/aip/jcp/article-abstract/131/10/105106/906179/Highly-efficient-energy-excitation-transfer-in?redirectedFrom=fulltext}
}

@article{Verstraete2009,
  title={Quantum computation and quantum-state engineering driven by dissipation},
  author={Verstraete, Frank and Wolf, Michael M and Ignacio Cirac, J},
  journal={Nature physics},
  volume={5},
  number={9},
  pages={633--636},
  year={2009},
  publisher={Nature Publishing Group UK London},
url={https://www.nature.com/articles/nphys1342}
}

@article{krauter2011,
  title = {Entanglement Generated by Dissipation and Steady State Entanglement of Two Macroscopic Objects},
  author = {Krauter, Hanna and Muschik, Christine A. and Jensen, Kasper and Wasilewski, Wojciech and Petersen, Jonas M. and Cirac, J. Ignacio and Polzik, Eugene S.},
  journal = {Phys. Rev. Lett.},
  volume = {107},
  issue = {8},
  pages = {080503},
  numpages = {5},
  year = {2011},
  month = {Aug},
  publisher = {American Physical Society},
  doi = {10.1103/PhysRevLett.107.080503},
  url = {https://link.aps.org/doi/10.1103/PhysRevLett.107.080503}
}

@article{zurek2003,
  title = {Decoherence, einselection, and the quantum origins of the classical},
  author = {Zurek, Wojciech Hubert},
  journal = {Rev. Mod. Phys.},
  volume = {75},
  issue = {3},
  pages = {715--775},
  numpages = {0},
  year = {2003},
  month = {May},
  publisher = {American Physical Society},
  doi = {10.1103/RevModPhys.75.715},
  url = {https://link.aps.org/doi/10.1103/RevModPhys.75.715}
}

@book{schlosshauer2007,
  title={Decoherence and the quantum-to-classical transition},
  author={Schlosshauer, Maximilian},
  year={2007},
  publisher={Springer},
url={https://link.springer.com/book/10.1007/978-3-540-35775-9}
}

@article{Holevo12,
doi = {10.1088/0034-4885/75/4/046001},
url = {https://doi.org/10.1088/0034-4885/75/4/046001},
year = {2012},
month = {mar},
publisher = {},
volume = {75},
number = {4},
pages = {046001},
author = {Holevo, A S and Giovannetti, V},
title = {Quantum channels and their entropic characteristics},
journal = {Reports on Progress in Physics}
}

@article{horodecki09,
  title = {Quantum entanglement},
  author = {Horodecki, Ryszard and Horodecki, Pawe\l{} and Horodecki, Micha\l{} and Horodecki, Karol},
  journal = {Rev. Mod. Phys.},
  volume = {81},
  issue = {2},
  pages = {865--942},
  numpages = {0},
  year = {2009},
  month = {Jun},
  publisher = {American Physical Society},
  doi = {10.1103/RevModPhys.81.865},
  url = {https://link.aps.org/doi/10.1103/RevModPhys.81.865}
}

@article{Guhne09,
title = {Entanglement detection},
journal = {Physics Reports},
volume = {474},
number = {1},
pages = {1-75},
year = {2009},
issn = {0370-1573},
doi = {https://doi.org/10.1016/j.physrep.2009.02.004},
url = {https://www.sciencedirect.com/science/article/pii/S0370157309000623},
author = {Otfried Gühne and Géza Tóth}
}

@article{Ekert91,
  title = {Quantum cryptography based on Bell's theorem},
  author = {Ekert, Artur K.},
  journal = {Phys. Rev. Lett.},
  volume = {67},
  issue = {6},
  pages = {661--663},
  numpages = {0},
  year = {1991},
  month = {Aug},
  publisher = {American Physical Society},
  doi = {10.1103/PhysRevLett.67.661},
  url = {https://link.aps.org/doi/10.1103/PhysRevLett.67.661}
}

@article{GRTZ02,
  title = {Quantum cryptography},
  author = {Gisin, Nicolas and Ribordy, Gr\'egoire and Tittel, Wolfgang and Zbinden, Hugo},
  journal = {Rev. Mod. Phys.},
  volume = {74},
  issue = {1},
  pages = {145--195},
  numpages = {0},
  year = {2002},
  month = {Mar},
  publisher = {American Physical Society},
  doi = {10.1103/RevModPhys.74.145},
  url = {https://link.aps.org/doi/10.1103/RevModPhys.74.145}
}

@article{Teleportation,
  title = {Teleporting an unknown quantum state via dual classical and Einstein-Podolsky-Rosen channels},
  author = {Bennett, C. H. and Brassard, G. and Cr\'epeau, C. and Jozsa, R. and Peres, A. and Wootters, W. K.},
  journal = {Phys. Rev. Lett.},
  volume = {70},
  issue = {13},
  pages = {1895},
  numpages = {0},
  year = {1993},
  month = {Mar},
  publisher = {American Physical Society},
   url = {https://link.aps.org/doi/10.1103/PhysRevLett.70.1895}
}

@Article{Hu2023,
author={Hu, Xiao-Min
and Guo, Yu
and Liu, Bi-Heng
and Li, Chuan-Feng
and Guo, Guang-Can},
title={Progress in quantum teleportation},
journal={Nature Reviews Physics},
year={2023},
month={Jun},
day={01},
volume={5},
number={6},
pages={339-353},
issn={2522-5820},
doi={10.1038/s42254-023-00588-x},
url={https://doi.org/10.1038/s42254-023-00588-x}
}

@article{Densed1,
  title = {Communication via one- and two-particle operators on Einstein-Podolsky-Rosen states},
  author = {Bennett, C. H. and Wiesner, S. J.},
  journal = {Phys. Rev. Lett.},
  volume = {69},
  issue = {20},
  pages = {2881},
  numpages = {0},
  year = {1992},
  month = {Nov},
  publisher = {American Physical Society},
  url = {https://link.aps.org/doi/10.1103/PhysRevLett.69.2881}
}

@article{SBDS19,
  title = {One-shot conclusive multiport quantum dense coding capacities},
  author = {Srivastava, Chirag and Bera, Anindita and Sen(De), Aditi and Sen, Ujjwal},
  journal = {Phys. Rev. A},
  volume = {100},
  issue = {5},
  pages = {052304},
  numpages = {11},
  year = {2019},
  month = {Nov},
  publisher = {American Physical Society},
  doi = {10.1103/PhysRevA.100.052304},
  url = {https://link.aps.org/doi/10.1103/PhysRevA.100.052304}
}

@article{Guo2019,
author = {Guo, Yu and Liu, Bi-Heng and Li, Chuan-Feng and Guo, Guang-Can},
title = {Advances in Quantum Dense Coding},
journal = {Advanced Quantum Technologies},
volume = {2},
number = {5-6},
pages = {1900011},
doi = {https://doi.org/10.1002/qute.201900011},
url = {https://advanced.onlinelibrary.wiley.com/doi/abs/10.1002/qute.201900011},
year = {2019}
}

@Article{Pironio2010,
author={Pironio, S.
and Ac{\'i}n, A.
and Massar, S.
and de la Giroday, A. Boyer
and Matsukevich, D. N.
and Maunz, P.
and Olmschenk, S.
and Hayes, D.
and Luo, L.
and Manning, T. A.
and Monroe, C.},
title={Random numbers certified by Bell's theorem},
journal={Nature},
year={2010},
month={Apr},
day={01},
volume={464},
number={7291},
pages={1021-1024},
issn={1476-4687},
doi={10.1038/nature09008},
url={https://doi.org/10.1038/nature09008}
}

@article{Christensen2013,
  title = {Detection-Loophole-Free Test of Quantum Nonlocality, and Applications},
  author = {Christensen, B. G. and McCusker, K. T. and Altepeter, J. B. and Calkins, B. and Gerrits, T. and Lita, A. E. and Miller, A. and Shalm, L. K. and Zhang, Y. and Nam, S. W. and Brunner, N. and Lim, C. C. W. and Gisin, N. and Kwiat, P. G.},
  journal = {Phys. Rev. Lett.},
  volume = {111},
  issue = {13},
  pages = {130406},
  numpages = {5},
  year = {2013},
  month = {Sep},
  publisher = {American Physical Society},
  doi = {10.1103/PhysRevLett.111.130406},
  url = {https://link.aps.org/doi/10.1103/PhysRevLett.111.130406}
}

@article{Bell64,
  title = {On the {E}instein {P}odolsky {R}osen paradox},
  author = {Bell, J. S.},
  journal = {Physics Physique Fizika},
  volume = {1},
  issue = {3},
  pages = {195--200},
  numpages = {6},
  year = {1964},
  month = {Nov},
  publisher = {American Physical Society},
  doi = {10.1103/PhysicsPhysiqueFizika.1.195},
  url = {https://link.aps.org/doi/10.1103/PhysicsPhysiqueFizika.1.195}
}

@book{bell_aspect_2004, 
place={Cambridge}, 
edition={2}, 
title={Speakable and Unspeakable in Quantum Mechanics: Collected Papers on Quantum Philosophy}, 
publisher={Cambridge University Press}, 
author={Bell, J. S.}, 
year={2004}
}

@article{bell-review,
  title = {{B}ell nonlocality},
  author = {Brunner, N. and Cavalcanti, D. and Pironio, S. and Scarani, V. and Wehner, S.},
  journal = {Rev. Mod. Phys.},
  volume = {86},
  issue = {2},
  pages = {419},
  numpages = {60},
  year = {2014},
  month = {Apr},
  publisher = {American Physical Society},
  url = {https://link.aps.org/doi/10.1103/RevModPhys.86.419}
}

@article{Schrödinger_1935, title={Discussion of Probability Relations between Separated Systems}, volume={31}, DOI={10.1017/S0305004100013554}, number={4}, journal={Mathematical Proceedings of the Cambridge Philosophical Society}, author={Schrödinger, E.}, year={1935}, pages={555–563}}

@article{Schrödinger_1936, title={Probability relations between separated systems}, volume={32}, DOI={10.1017/S0305004100019137}, number={3}, journal={Mathematical Proceedings of the Cambridge Philosophical Society}, author={Schrödinger, E.}, year={1936}, pages={446–452}}

@article{Wiseman07,
  title = {Steering, Entanglement, Nonlocality, and the Einstein-Podolsky-Rosen Paradox},
  author = {Wiseman, H. M. and Jones, S. J. and Doherty, A. C.},
  journal = {Phys. Rev. Lett.},
  volume = {98},
  issue = {14},
  pages = {140402},
  numpages = {4},
  year = {2007},
  month = {Apr},
  publisher = {American Physical Society},
  doi = {10.1103/PhysRevLett.98.140402},
  url = {https://link.aps.org/doi/10.1103/PhysRevLett.98.140402}
}

@article{Uola20,
  title = {Quantum steering},
  author = {Uola, Roope and Costa, Ana C. S. and Nguyen, H. Chau and G\"uhne, Otfried},
  journal = {Rev. Mod. Phys.},
  volume = {92},
  issue = {1},
  pages = {015001},
  numpages = {40},
  year = {2020},
  month = {Mar},
  publisher = {American Physical Society},
  doi = {10.1103/RevModPhys.92.015001},
  url = {https://link.aps.org/doi/10.1103/RevModPhys.92.015001}
}

@article{Me2,
  title = {Finding optimal strategies for minimum-error quantum-state discrimination},
  author = {Je\ifmmode \check{z}\else \v{z}\fi{}ek, M. and \ifmmode \check{R}\else \v{R}\fi{}eh\'a\ifmmode \check{c}\else \v{c}\fi{}ek, J. and Fiur\'a\ifmmode \check{s}\else \v{s}\fi{}ek, J.},
  journal = {Phys. Rev. A},
  volume = {65},
  issue = {6},
  pages = {060301},
  year = {2002},
  month = {Jun},
  publisher = {American Physical Society},
  url = {https://link.aps.org/doi/10.1103/PhysRevA.65.060301}
}

@article{Me4,
  title = {Distinguishing mixed quantum states: Minimum-error discrimination versus optimum unambiguous discrimination},
  author = {Herzog, U. and Bergou, J. A.},
  journal = {Phys. Rev. A},
  volume = {70},
  pages = {022302},
  year = {2004},
  month = {Aug},
  publisher = {American Physical Society},
  url = {https://link.aps.org/doi/10.1103/PhysRevA.70.022302}
}

@article{Me6,
  title={On the conditions for discrimination between quantum states with minimum error},
  author={Barnett, S. M.  and Croke, S. },
  journal={arXiv preprint arXiv:0810.1919},
  year={2008},
url={
https://doi.org/10.1088/1751-8113/42/6/062001
}
}

@article{Me7,
  title = {Minimum-error discrimination of quantum states: Bounds and comparisons},
  author = {Qiu, D. and Li, L.},
  journal = {Phys. Rev. A},
  volume = {81},
  pages = {042329},
  year = {2010},
  month = {Apr},
  publisher = {American Physical Society},
  doi = {10.1103/PhysRevA.81.042329},
  url = {https://link.aps.org/doi/10.1103/PhysRevA.81.042329}
}

@article{Me12,
  title={Structure of minimum-error quantum state discrimination},
  author={Bae, J. },
  journal={New Journal of Physics},
  volume={15},
  number={7},
  pages={073037},
  year={2013},
  publisher={IOP Publishing},
url={
https://doi.org/10.1088/1367-2630/15/7/073037
Focus to learn more
}
}

@article{Bae_2015,
doi = {10.1088/1751-8113/48/8/083001},
url = {https://doi.org/10.1088/1751-8113/48/8/083001},
year = {2015},
month = {jan},
publisher = {IOP Publishing},
volume = {48},
number = {8},
pages = {083001},
author = {Bae, Joonwoo and Kwek, Leong-Chuan},
title = {Quantum state discrimination and its applications},
journal = {Journal of Physics A: Mathematical and Theoretical}
}

@article{Me16,
  title = {Minimal-error quantum state discrimination versus robustness of entanglement: More indistinguishability with less entanglement},
  author = {Saha, Debarupa and Sen, Kornikar and Srivastava, Chirag and Sen, Ujjwal},
  journal = {Phys. Rev. A},
  volume = {112},
  issue = {6},
  pages = {062441},
  numpages = {16},
  year = {2025},
  month = {Dec},
  publisher = {American Physical Society},
  doi = {10.1103/dy3b-lyw8},
  url = {https://link.aps.org/doi/10.1103/dy3b-lyw8}
}

@article{helstrom1,
  title={Quantum detection and estimation theory},
  author={Carl W. Helstrom},
  journal={Journal of Statistical Physics},
  year={1969},
  volume={1},
  pages={231-252},
  url={https://api.semanticscholar.org/CorpusID:12758217}
}

@article{Halder20,
  title = {Locally distinguishing quantum states with limited classical communication},
  author = {Halder, Saronath and Srivastava, Chirag},
  journal = {Phys. Rev. A},
  volume = {101},
  issue = {5},
  pages = {052313},
  numpages = {7},
  year = {2020},
  month = {May},
  publisher = {American Physical Society},
  doi = {10.1103/PhysRevA.101.052313},
  url = {https://link.aps.org/doi/10.1103/PhysRevA.101.052313}
}

@book{Breuer07,
  title = {The Theory of Open Quantum Systems},
  ISBN = {9780191706349},
  url = {http://dx.doi.org/10.1093/acprof:oso/9780199213900.001.0001},
  DOI = {10.1093/acprof:oso/9780199213900.001.0001},
  publisher = {Oxford University PressOxford},
  author = {Breuer,  Heinz-Peter and Petruccione,  Francesco},
  year = {2007},
  month = jan 
}

@book{Wiseman09, place={Cambridge}, title={Quantum Measurement and Control}, publisher={Cambridge University Press}, author={Wiseman, Howard M. and Milburn, Gerard J.}, year={2009}}

@article{Devetak05,
  title = {The Capacity of a Quantum Channel for Simultaneous Transmission of Classical and Quantum Information},
  volume = {256},
  ISSN = {1432-0916},
  url = {http://dx.doi.org/10.1007/s00220-005-1317-6},
  DOI = {10.1007/s00220-005-1317-6},
  number = {2},
  journal = {Communications in Mathematical Physics},
  publisher = {Springer Science and Business Media LLC},
  author = {Devetak,  I. and Shor,  P. W.},
  year = {2005},
  month = mar,
  pages = {287–303}
}

@article{Silva15,
  title = {Multiple Observers Can Share the Nonlocality of Half of an Entangled Pair by Using Optimal Weak Measurements},
  author = {Silva, Ralph and Gisin, Nicolas and Guryanova, Yelena and Popescu, Sandu},
  journal = {Phys. Rev. Lett.},
  volume = {114},
  issue = {25},
  pages = {250401},
  numpages = {5},
  year = {2015},
  month = {Jun},
  publisher = {American Physical Society},
  doi = {10.1103/PhysRevLett.114.250401},
  url = {https://link.aps.org/doi/10.1103/PhysRevLett.114.250401}
}

@article{Brown20,
  title = {Arbitrarily Many Independent Observers Can Share the Nonlocality of a Single Maximally Entangled Qubit Pair},
  author = {Brown, Peter J. and Colbeck, Roger},
  journal = {Phys. Rev. Lett.},
  volume = {125},
  issue = {9},
  pages = {090401},
  numpages = {5},
  year = {2020},
  month = {Aug},
  publisher = {American Physical Society},
  doi = {10.1103/PhysRevLett.125.090401},
  url = {https://link.aps.org/doi/10.1103/PhysRevLett.125.090401}
}

@InProceedings{Curchod17,
  author =	{Curchod, Florian J. and Johansson, Markus and Augusiak, Remigiusz and Hoban, Matty J. and Wittek, Peter and Ac{\'\i}n, Antonio},
  title =	{{A Single Entangled System Is an Unbounded Source of Nonlocal Correlations and of Certified Random Numbers}},
  booktitle =	{12th Conference on the Theory of Quantum Computation, Communication and Cryptography (TQC 2017)},
  pages =	{1:1--1:23},
  series =	{Leibniz International Proceedings in Informatics (LIPIcs)},
  ISBN =	{978-3-95977-034-7},
  ISSN =	{1868-8969},
  year =	{2018},
  volume =	{73},
  editor =	{Wilde, Mark M.},
  publisher =	{Schloss Dagstuhl -- Leibniz-Zentrum f{\"u}r Informatik},
  address =	{Dagstuhl, Germany},
  URL =		{https://drops.dagstuhl.de/entities/document/10.4230/LIPIcs.TQC.2017.1},
  URN =		{urn:nbn:de:0030-drops-85809},
  doi =		{10.4230/LIPIcs.TQC.2017.1}
}

@article{SPS22,
  title = {Entanglement witnessing by arbitrarily many independent observers recycling a local quantum shared state},
  author = {Srivastava, Chirag and Pandit, Mahasweta and Sen, Ujjwal},
  journal = {Phys. Rev. A},
  volume = {105},
  issue = {6},
  pages = {062413},
  numpages = {6},
  year = {2022},
  month = {Jun},
  publisher = {American Physical Society},
  doi = {10.1103/PhysRevA.105.062413},
  url = {https://link.aps.org/doi/10.1103/PhysRevA.105.062413}
}

@article{PSS22,
  title = {Recycled entanglement detection by arbitrarily many sequential and independent pairs of observers},
  author = {Pandit, Mahasweta and Srivastava, Chirag and Sen, Ujjwal},
  journal = {Phys. Rev. A},
  volume = {106},
  issue = {3},
  pages = {032419},
  numpages = {6},
  year = {2022},
  month = {Sep},
  publisher = {American Physical Society},
  doi = {10.1103/PhysRevA.106.032419},
  url = {https://link.aps.org/doi/10.1103/PhysRevA.106.032419}
}

@article{SPS25,
  title = {Unbounded recycling of genuine multiparty entanglement for any number of qubits},
  author = {Srivastava, Chirag and Pandit, Mahasweta and Sen, Ujjwal},
  journal = {Phys. Rev. A},
  volume = {111},
  issue = {1},
  pages = {012413},
  numpages = {12},
  year = {2025},
  month = {Jan},
  publisher = {American Physical Society},
  doi = {10.1103/PhysRevA.111.012413},
  url = {https://link.aps.org/doi/10.1103/PhysRevA.111.012413}
}

@article{MSSS25,
  title = {Local entanglement transfer from an entanglement source to multiple pairs of spatially separated observers},
  author = {Mondal, Tanmoy and Sen, Kornikar and Srivastava, Chirag and Sen, Ujjwal},
  journal = {Phys. Rev. A},
  volume = {112},
  issue = {1},
  pages = {L010402},
  numpages = {6},
  year = {2025},
  month = {Jul},
  publisher = {American Physical Society},
  doi = {10.1103/17s4-rtdm},
  url = {https://link.aps.org/doi/10.1103/17s4-rtdm}
}

@article{Branciard10,
  title = {Characterizing the Nonlocal Correlations Created via Entanglement Swapping},
  author = {Branciard, C. and Gisin, N. and Pironio, S.},
  journal = {Phys. Rev. Lett.},
  volume = {104},
  issue = {17},
  pages = {170401},
  numpages = {4},
  year = {2010},
  month = {Apr},
  publisher = {American Physical Society},
  doi = {10.1103/PhysRevLett.104.170401},
  url = {https://link.aps.org/doi/10.1103/PhysRevLett.104.170401}
}

@article{Branciard12,
  title = {Bilocal versus nonbilocal correlations in entanglement-swapping experiments},
  author = {Branciard, Cyril and Rosset, Denis and Gisin, Nicolas and Pironio, Stefano},
  journal = {Phys. Rev. A},
  volume = {85},
  issue = {3},
  pages = {032119},
  numpages = {21},
  year = {2012},
  month = {Mar},
  publisher = {American Physical Society},
  doi = {10.1103/PhysRevA.85.032119},
  url = {https://link.aps.org/doi/10.1103/PhysRevA.85.032119}
}

@article{Fritz12,
doi = {10.1088/1367-2630/14/10/103001},
url = {https://doi.org/10.1088/1367-2630/14/10/103001},
year = {2012},
month = {oct},
publisher = {IOP Publishing},
volume = {14},
number = {10},
pages = {103001},
author = {Fritz, Tobias},
title = {Beyond Bell's theorem: correlation scenarios},
journal = {New Journal of Physics}
}

@article{Boreiri23,
  title = {Towards a minimal example of quantum nonlocality without inputs},
  author = {Boreiri, Sadra and Girardin, Antoine and Ulu, Bora and Lipka-Bartosik, Patryk and Brunner, Nicolas and Sekatski, Pavel},
  journal = {Phys. Rev. A},
  volume = {107},
  issue = {6},
  pages = {062413},
  numpages = {12},
  year = {2023},
  month = {Jun},
  publisher = {American Physical Society},
  doi = {10.1103/PhysRevA.107.062413},
  url = {https://link.aps.org/doi/10.1103/PhysRevA.107.062413}
}

@article{Berta10,
  title = {The uncertainty principle in the presence of quantum memory},
  volume = {6},
  ISSN = {1745-2481},
  url = {http://dx.doi.org/10.1038/nphys1734},
  DOI = {10.1038/nphys1734},
  number = {9},
  journal = {Nature Physics},
  publisher = {Springer Science and Business Media LLC},
  author = {Berta,  Mario and Christandl,  Matthias and Colbeck,  Roger and Renes,  Joseph M. and Renner,  Renato},
  year = {2010},
  month = July,
  pages = {659–662}
}

@article{GREENBERGER88,
title = {Simultaneous wave and particle knowledge in a neutron interferometer},
journal = {Physics Letters A},
volume = {128},
number = {8},
pages = {391-394},
year = {1988},
issn = {0375-9601},
doi = {https://doi.org/10.1016/0375-9601(88)90114-4},
url = {https://www.sciencedirect.com/science/article/pii/0375960188901144},
author = {Daniel M. Greenberger and Allaine Yasin}
}

@article{Basso2022,
  title = {Predictability as a quantum resource},
  volume = {21},
  ISSN = {1573-1332},
  url = {http://dx.doi.org/10.1007/s11128-022-03503-y},
  DOI = {10.1007/s11128-022-03503-y},
  number = {5},
  journal = {Quantum Information Processing},
  publisher = {Springer Science and Business Media LLC},
  author = {Basso,  Marcos L. W. and Maziero,  Jonas},
  year = {2022},
  month = May 
}

@article{Martinez26,
  title = {From top quarks to enhanced quantum key distribution: A framework for optimal predictability of quantum observables},
  author = {Mart\'{\i}nez-Moreno, Dennis I. and Castillo-Celeita, Miguel and Bussandri, Diego G.},
  journal = {Phys. Rev. Appl.},
  volume = {25},
  issue = {1},
  pages = {014063},
  numpages = {21},
  year = {2026},
  month = {Jan},
  publisher = {American Physical Society},
  doi = {10.1103/vw1p-tzdn},
  url = {https://link.aps.org/doi/10.1103/vw1p-tzdn}
}

@article{Matthews10,
  title = {Entanglement in channel discrimination with restricted measurements},
  author = {Matthews, William and Piani, Marco and Watrous, John},
  journal = {Phys. Rev. A},
  volume = {82},
  issue = {3},
  pages = {032302},
  numpages = {8},
  year = {2010},
  month = {Sep},
  publisher = {American Physical Society},
  doi = {10.1103/PhysRevA.82.032302},
  url = {https://link.aps.org/doi/10.1103/PhysRevA.82.032302}
}

@article{Piani15,
  title = {Necessary and Sufficient Quantum Information Characterization of Einstein-Podolsky-Rosen Steering},
  author = {Piani, Marco and Watrous, John},
  journal = {Phys. Rev. Lett.},
  volume = {114},
  issue = {6},
  pages = {060404},
  numpages = {6},
  year = {2015},
  month = {Feb},
  publisher = {American Physical Society},
  doi = {10.1103/PhysRevLett.114.060404},
  url = {https://link.aps.org/doi/10.1103/PhysRevLett.114.060404}
}

\end{document}